\documentclass[a4paper,UKenglish,cleveref, autoref, thm-restate,final]{lipics-v2019}

\title{AWLCO: All-Window Length Co-Occurrence}

\titlerunning{AWLCO: All-Window Length Co-occurrence}

\author{Joshua Sobel}{Department of Computer Science, University of Rochester }{jsobel3@u.rochester.edu}{}{}
\author{Noah Betram}{Department of Mathematics, University of Rochester }{nbertram@u.rochester.edu}{}{}
\author{Chen Ding}{Department of Computer Science, University of Rochester }{cding@cs.rochester.edu}{}{}
\author{Fatemeh Nargesian}{Department of Computer Science, University of Rochester }{fnargesian@cs.rochester.edu}{}{}
\author{Daniel Gildea}{Department of Computer Science, University of Rochester }{gildea@cs.rochester.edu}{}{}

\authorrunning{J. Sobel, N. Bertram, C. Ding, F. Nargesian, D. Gildea} 



\nolinenumbers 

\hideLIPIcs  

\usepackage{amsmath}
\usepackage{amsthm}
\usepackage{mathtools}
\usepackage{bbm}
\usepackage{tikz}
\usetikzlibrary{shapes.multipart}
\usepackage{graphicx}
\usepackage{pgfplots}
\pgfplotsset{width=4cm,compat=1.9}
\usetikzlibrary{patterns,matrix}
\usepackage[ruled,vlined]{algorithm2e}

\usepackage{textcomp}
\usepackage{epstopdf}
\usepackage{tabularx}
\usepackage{ragged2e}
\usepackage[labelfont=bf]{caption}
\usepackage[font=small,labelfont=bf,tableposition=top]{caption}

\DeclareMathOperator{\co}{{ \it co-occurrence}}
\newcommand{\tesla}{\ensuremath{\mathrm{tesla}}}
\newcommand{\btesla}{\ensuremath{\mathrm{btesla}}}

\newcommand{\name}{{\tt AWLCO}}
\DeclareMathOperator{\Tr}{\text{Tr}}

\DeclarePairedDelimiter\floor{\lfloor}{\rfloor}

\newcommand{\eat}[1]{}

\newenvironment{proofsketch}[1][Proof Sketch.]{\par
  \pushQED{\qed}%
  \normalfont 
  \trivlist
  \item[\hskip\labelsep
        \color{lipicsGray}\sffamily\bfseries
    #1
    ]\ignorespaces}{
  \popQED\endtrivlist
}

\begin{document}

\maketitle

\begin{abstract}
    Analyzing patterns in a sequence of events has applications in text analysis, computer programming, and genomics research.  
    In this paper, we consider the {\em all-window-length} analysis model which analyzes a sequence of events with respect to windows of all lengths. 
    We study the exact co-occurrence counting problem for the all-window-length analysis model.  
    Our first algorithm is an offline algorithm that  
    counts all-window-length co-occurrences by performing 
    multiple passes over a sequence and computing single-window-length co-occurrences.  
    This algorithm has the time complexity $O(n)$ for each window length and thus a total complexity of $O(n^2)$ and the space complexity $O(|I|)$ for 
    a sequence of size $n$ and an itemset of size $|I|$. 
    We propose \name{}, an online algorithm that computes 
    all-window-length co-occurrences in a single pass with the expected time complexity of $O(n)$ and space complexity of $O(\sqrt{n|I|})$. Following this, we generalize our use case to patterns in which we propose an algorithm that computes all-window-length co-occurrence with expected time complexity $O(n|I|)$ and space complexity $O(\sqrt{n|I|} + e_{max}|I|)$, where $e_{max}$ is the length of the largest pattern.
\end{abstract}

\section{Introduction}
\label{sec:intro}

Analyzing regularities in streams and event sequences has applications in data analytics as well as programming languages, natural language processing, and genomics. 
Examples of a event sequence include a sequence of system logs, memory requests by a program, tweets by a user, 
a series of symptoms, a sequence of words in a document, or an RNA sequence.  
One metric of regularity is {\em co-occurrence}~\cite{Liu:2018:TWC:3291291.3291322,YuY0LP15} --- the number of times that an entire set of items or more broadly of patterns is contained within a sliding window of an arbitrary size. 
For example, 
consider the  sequence ``$abccba$'' 
and window size three. This sequence of events  contains four such windows: ``$abc$'', ``$bcc$'', ``$ccb$'', and ``$cba$''. 
We see that both ``$a$'' and ``$b$''  appear together in two windows. 
Thus, itemset $\{a, b\}$ co-occurs twice for window size of three.  In the sequence ``\emph{cat dog cat}'' with window size seven, we see that the words, referred to as patterns, ``$cat$'' and ``$dog$'' both appear as substrings in two windows, and thus the pattern set  $\{cat, dog\}$ have co-occurrence of two with window size seven.  

Most applications assume that the window is given by a user or defined in an adhoc manner.  
Existing counting algorithms for streams often assume the {\em sliding-window} model of computation, that is answering queries or mining is done over the last $w$ most recent data elements~\cite{DBLP:books/sp/16/DatarM16,DBLP:journals/siamcomp/DatarGIM02}.  
Successful pattern-searching tools, such as {\em ShapeSearch}, enable the search for desired patterns within a fixed window size in trendlines~\cite{SiddiquiLWKP20}. 
However, in certain applications of co-occurrence analysis, 
the query is about identifying the time windows that satisfy certain conditions on the co-occurrence.  
For instance, in text analysis, what is the time window in which 
a set of events are very likely to appear? 
Or, at which time window does the co-occurrence of a set of words in a document become random? 
Or, how often do two or multiple gene expression patterns co-occur in an RNA sequence?  
These applications require the analysis of all possible window lengths, possibly as large as the size of the sequence.  

{\bf The All-Window-Length Analysis Model}
In this paper, we consider a new analysis model of computation for streams and sequences, the {\em all-window-length analysis model}, where the analysis of a sequence of data elements is done in {\em one pass} for all window lengths,  starting from the size of a pattern up to the size of a sequence. Unlike single-window-length analysis, in this model, window length becomes a variable. 
We consider the co-occurrence counting of items and patterns in this analysis model. 
A pattern is a string with characters drawn from alphabet $\mathbb{A}$. 
Given a sequence $T$ of size $n$, and an itemset $I$ consisting of patterns, find the the number of windows in which every pattern in $I$ occurs for all window lengths $x\in\{1, \ldots, n\}$ in  $T$. 
This model enables us to perform analysis without apriori knowledge of window-size, 
i.e. a window size can be chosen and analyzed on demand at query time. 
For a sequence $T$ of size $n$ and an itemset $I$ consisting of $|I|$ unique tokens, the co-occurrence analysis considers $\sum_{x=1}^{n}(n - x + 1)$ windows. 
We propose efficient exact algorithms and theoretical analysis for the co-occurrence counting of sets of items and patterns under  this  analysis model. 
Note that this analysis model is different than the setting of counting frequent itemset in a stream, in which data elements arrive in baskets of arbitrary lengths and the goal is to find the itemset that appears in $s$ fraction of the baskets, where $s$ is a support threshold~\cite{rajaraman2014mining,MankuM02,AgrawalS94,AgrawalIS93}. 

{\bf Applications} We expect the all-window-length analysis model to open research  opportunities that lead to solving problems in natural language processing, the optimization of the memory layout of programs, and accelerating the search for RNA sequences in genomes. 
In natural language processing, the co-occurrence of words within a sliding window is the basis for training word embeddings,
which are vector representations of a word's meaning and usage
\cite{mikolov2013distributed,pennington-socher-manning:2014:EMNLP2014}. 
Different window sizes are useful for different purposes;
embeddings derived from
smaller windows tend to represent syntactic information while
larger windows represent semantic information~\cite{Schutze97}.
Identifying an effective window length for training word embeddings 
requires the efficient exploration of the relationship between window size and co-occurrence frequency of words~\cite{LevyG14}. 

The application of all-window-length co-occurrence analysis in programming languages  is in the optimization of the  memory layout of programs.  Modern processor performance is dependent on cache performance and cache block utilization. A set of data elements belong to the same {\em affinity group} if they are always accessed close to each other.  This closeness is defined by $k$-linkedness. A reference affinity forms a unique partition of data for every $k$, and the relation between different $k$s is hierarchical, meaning the affinity groups at link length $k$ are a finer partition of the groups at $k+1$. Reference affinity has been used to optimize the memory layout in data structure splitting~\cite{Zhong+:PLDI04}, whole-program code layout~\cite{Lavaee+:CC19}, and both~\cite{Zhang+:POPL06}. Finding affinity groups requires the analysis of the access co-occurrence of data elements in memory access traces for all $k$s. 

Research has shown that analyzing nucleotide co-occurrence over the entire human genome provides a powerful insight into the evolution of viruses~\cite{Shapiroe,PMID:18032723}. Co-occurrence is a method for tracking cooperative genomic interactions as a major force underlying virus evolution. Existing co-occurrence network construction tools such as {\tt cooccurNet}~\cite{ZouWDWWLJP17} consider pairs of nucleotides or amino acids for analysis and apply filters on the significance of the co-occurrence of genes. The distance in a co-occurrence network counts for the relatedness of genes. An all-window-length analysis of the co-occurrence gene sequences provides further insight into pattern analysis in genomics.

{\bf Results} In this paper, we propose an efficient algorithm that computes all-window-length exact co-occurrence of patterns in a single pass. 
For co-occurrence of itemsets of size one or two, 
our past work proposed a linear time algorithm (in sequence length) to compute co-occurrence 
for all window lengths~\cite{Liu:2018:TWC:3291291.3291322}.
To analyze co-occurrence, first, we introduce an algorithm to calculate co-occurrence that runs in $O(n)$ time, is easily understood, and uses $O(|I|)$ space for single-window-length co-occurrence, where $n$ is the length of the sequence, and $I$ is the set of co-occurring items.  
However, to find the co-occurrence across all window lengths the algorithm would require to compute the co-occurrence for each window length separately and use $O(n^2)$ time which is impractical for large datasets.  

We propose \name{}, a time- and space-efficient algorithm that computes the exact co-occurrence of itemsets for all window lengths, 
in a single pass.  The algorithm computes co-occurrence by finding \textit{gaps} in the sequence, or substrings of the sequence that do not contain subsets of the queried pattern.  This is a novel approach to compute co-occurrence and provides an improved algorithm since the stored gaps are not bound to any window lengths, thus, the collection of gaps allows the co-occurrence to be determined for all window lengths in a single pass through the gaps.  
Furthermore, we propose a simple approach  for computing all of the gaps for an itemset in a single pass through the sequence.  The relevant gaps can be found by iterating through the sequence and keeping track of the items and the orders they last appeared.  We theoretically prove that gaps are only relevant and counted if the current item encountered in the sequence is the  item that was seen furthest in the past, thus,  drastically reducing the amount of space and updates needed. 
\name{}  enables all-window-length queries in expected $O(n)$ time by using $O(\sqrt{n|I|})$ additional space, assuming a perfect hashing function. 

Finally, we generalize our problem to finding the co-occurrence of a set of patterns. We argue that finding an algorithm that handles multiple elements at the same index of a sequence would solve all window length pattern co-occurrence. We present an algorithm for pattern co-occurrence counting with the expected time complexity $O(n|I|)$ and space complexity  $O(\sqrt{n|I|} + e_{\max}|I|)$, where $e_{\max}$ is the length of the largest pattern.

\section{Problem Definition}
\label{sec:problemdef}
We begin by fixing a vocabulary $\mathbb{A}$ that we will be working in. 
Let $T$ be a sequence with elements in $\mathbb{A}$. Sequence $T$ can be considered as a stream. Let $n$ be the length of the sequence $T$ and for any natural number $l$, let $[l] = \{1,\ldots,l\}$.  A sequence will have its indices zero indexed, i.e. $T[0]$ is the first element that appeared in the sequence and $T[i]$ is the element that  appeared at position
$i$.  We use $T[i\ldots j]$ to denote a sub-string of $T$. For example, $T[0\ldots j]$  indicates the first $j+1$ elements of sequence $T$.  
An itemset $I$ is a finite non-empty subset of $\mathbb{A}$.  For a sequence $T$, a window is a sub-string of $T$, or a contiguous selection of elements of $T$.  For sequence $T$ we define the window at index $i$ of length $x$ where $x\leq i+1$, \,$\omega(T, i, x)$, to be the window containing the $i$-th element of $T$ and the $x-1$ previous elements of $T$.  When it is clear what sequence is being referenced we will refer simply to $\omega(i, x)$.  For example, for the sequence $T=$``$abcdef$'', 
$\omega(3, 3)$ is ``$bcd$''.  
We define the co-occurrence count as the number of windows of length $x$ in sequence $T$ that contain all elements of the itemset $I$. 

\begin{definition}
{\em Single-window length co-occurrence problem}:  Given a sequence $T$ and an itemset $I$, find the co-occurrence count of itemset $I$ in windows of length $x$ in sequence $T$.
\begin{multline}\label{eq:co-def}
      \co\left(T, I, x \right) = |\{\omega(i, x) \,:\, i \in \{x-1, \dots, n-1\}, \forall e \in I, e \in \omega(i, x) \}|
\end{multline}
\end{definition}

\begin{example}
\label{ex:tesla}
\sloppy{Consider the sequence $T$=``abcabe''. The co-occurrence count of itemset $\{a,b\}$ in all windows with size four, 
$\co
\left(abcabe, \{a, b\}, 4 \right)$,  is three. }
\end{example}

In this paper, we consider the new  problem of finding co-occurrence counts of $I$ in $T$ for all window lengths. 

\begin{definition}
{\em All-window length co-occurrence problem}:  Given a sequence $T$ of size $n$, and an itemset $I$, find the co-occurrence counts of itemset $I$ in all windows of lengths $x\in\{|I|, \ldots, n\}$ in sequence $T$.
\end{definition}
 
In Section~\ref{sec:single-w-alg}, we define a baseline algorithm for finding all window length co-occurrence counts based on finding the single window length co-occurrence count. 
In Section~\ref{sec:all-w-alg}, we describe our algorithm for simultaneously finding co-occurrence counts of all window lengths in expected linear time in the length of the sequence and the space complexity of $O(\sqrt{n|I|})$. 

A pattern is a string with characters drawn from alphabet $\mathbb{A}$. A pattern $e$'s $i $th component is denoted $e[j]$ and the length of the pattern is $|e|$.  A pattern occurs in a sequence $T$ if there exists $j \in \{0,...,n\}$ such that for all $i \in \{0,\ldots ,|e| -1\}$, $T[j + i] = e[i]$. 

\begin{definition}
{\em All-window length pattern co-occurrence problem}:  Given a sequence $T$ of length $n$, and an itemset $I$ consisting of patterns, find the the number of windows in which every pattern in $I$ occurs for all window lengths $x\in\{1, \ldots, n\}$ in sequence $T$.
\end{definition}

\section{Single-Window-Length Co-occurrence}
\label{sec:single-w-alg}

    Consider an item $e\in\mathbb{A}$ and a sequence $T$. The time elapsed since last access of $e$ at index $i$, $\tesla(T,e,i)$, is the difference between $i$ and the greatest index where $e$ occurs in $T$ up to and possibly including $i$, and in the case that there is no occurrence of $e$ in the interval up to $i$ we define it to be $\infty$.  When the choice of $T$ is clear we use the shorthand $\tesla(e,i)$ instead. 
There is a direct connection between the \tesla{} values for items in the itemset and the number of times the items of the itemset co-occur. 
\begin{lemma}Itemset $I$ co-occurs in a window $\omega(i,x)$ if and only if $\max\{\tesla(e,i) | e\in I\}<x$.
\begin{proof}
The statement implies that for each $e \in I$, $\tesla(e, i) < x$, which implies that $e \in \omega(i,x)$.  Conversely, if each $e \in \omega(i,x)$, then we have $\tesla(e, i) < x$; therefore, we have \\ $\max\{tesla(e,i) | e\in I\}<x$.
\end{proof}
\end{lemma}

\begin{example}
\label{ex:stack} 
Consider the sequence $T=$``$abcabe$'' and itemset $\{a,b\}$. Suppose we have processed $T[0\ldots 3]$ and we know  $\tesla(a,3)=0$ and $\tesla(b,3)=2$. 
Since the max tesla value is two, the itemset does not co-occur in the size two window $\omega(3,2)$.
\end{example}

By the lemma, the co-occurrence defined in Equation~\ref{eq:co-def} can be computed by iterating through each index of the sequence and counting the number of times $\max\{\tesla(e,i)| e\in I\} < x$.  
\begin{multline}
\label{eq:coocc-tesla}
    \co\left(T, I, x \right) = |\{i \in \{x-1, \dots, n-1\} : 
    \max\{\tesla(e,i) | e\in I\}<x \}|
\end{multline}

\noindent{\bf Book Stack.} We now wish to have a systematic way of ordering items according to their corresponding time elapsed since last access. Let $Q$ denote the set of non-empty subsets of $\mathbb{A}$. Let $A$ be some element of $Q$ and suppose that $A=\{e^1,e^2,\ldots, e^{|A|}\}$, and they are labeled in such a way that at index $i$ in our sequence,
\[
    \tesla(e^1,i-1) \leq \tesla(e^2,i-1) \leq \cdots \leq \tesla(e^{|A|},i-1).
\] 
Now let $r_{i}^{j}: Q \to \mathbb{A}$ be given by $r_{i}^{j}(A) = e^j$, for $j \in \{1,...,|A|\}$.
That is to say that, $r$ arranges the members of $A$ in a finite sequence according to $\tesla(\cdot,i-1)$. This notation is robust as it allows for weak ordering and will be used to consider a generalized case later on. 
We call the realization of $r_{i}^{j}$ a {\em book-stack}, i.e. $\mathcal{S}_i = [(r_{i}^{1}(I), \tesla(r_{i}^{1}(I),i)), \ldots, (r_{i}^{|I|}(I), \tesla(r_{i}^{|I|}(I),i))]$ based on the above ordering, given a set $A$. 
We define $\mathcal{S}_i.{\tt retrieve}(j) = \tesla(r_{i}^{j},i)$. 
We define $\mathcal{S}_i.{\tt find} : A \rightarrow \{1, \ldots, |A|\}$, such that ${\tt find}(a) = j$, where $r_{i}^{j}(A) = a$. 
We define $\mathcal{S}_i.{\tt update}: \{1, \ldots, |A|\} \rightarrow \times_{l = 1}^{|A|}A$, in which $\mathcal{S}_i.{\tt update}(j) = (r_{i+1}^{1}(A),\ldots,r_{i+1}^{|A|}(A))$, where we have 
\[
r_{i+1}^l (A) = 
\begin{cases}
r_{i}^{j}(A),\ l = 1\\
r_{i}^{l+1}(A),\ 1 \leq l < j  \\
r_{i}^{j}(A),\ j < l \leq  |A|. \\
\end{cases}
\]
We therefore define $\mathcal{S}_{i+1} = \mathcal{S}_i.{\tt update}(\mathcal{S}_i.{\tt find}(T[i]))$.
It is straightforward to see that the ${\tt update}$ guarantees the correct ordering for $r_{i+1}^{j}$ based on $\tesla(\cdot,i)$. 
Figure~\ref{fig:stack-fig} illustrates ${\tt update}$ to a book-stack data structure step by step. 
By an abuse of notation, in our algorithms we refer to  $\mathcal{S}_i$ with $\mathcal{S}$. 

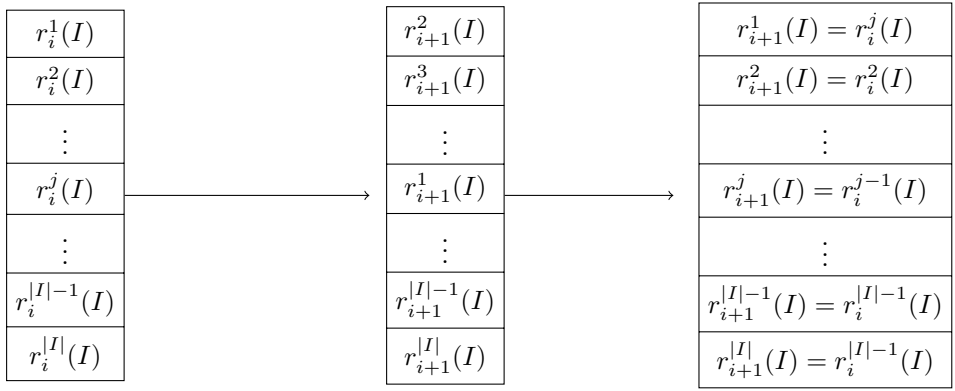
\begin{figure}[t]
\centering
\begin{tikzpicture}[stack/.style={rectangle split, rectangle split parts=#1,draw, anchor=center}]
\node[stack=7] at (-5,0) (s1) {
\nodepart{one}$r_i^1(I)$
\nodepart{two}$r_i^2(I)$
\nodepart{three}$\vdots$
\nodepart{four}$r_i^{j}(I)$
\nodepart{five}$\vdots$
\nodepart{six}$r_i^{|I|-1}(I)$
\nodepart{seven}$r_i^{|I|}(I)$
};

\node[stack=7] at (0,0) (s2) {
\nodepart{one}$r_{i+1}^2(I)$
\nodepart{two}$r_{i+1}^3(I)$
\nodepart{three}$\vdots$
\nodepart{four}$r_{i+1}^1(I)$
\nodepart{five}$\vdots$
\nodepart{six}$r_{i+1}^{|I|-1}(I)$
\nodepart{seven}$r_{i+1}^{|I|}(I)$
};

\node[stack=7] at (5,0) (s3) {
\nodepart{one}$r_{i+1}^1(I)=r_i^{j}(I)$
\nodepart{two}$r_{i+1}^2(I)=r_i^2(I)$
\nodepart{three}$\vdots$
\nodepart{four}$r_{i+1}^{j}(I)=r_{i}^{j-1}(I)$
\nodepart{five}$\vdots$
\nodepart{six}$r_{i+1}^{|I|-1}(I)=r_{i}^{|I|-1}(I)$
\nodepart{seven}$r_{i+1}^{|I|}(I)=r_{i}^{|I|-1}(I)$
};
\draw[=] (s1) (s2);
\draw[->] (s1) edge (-1,0);
\draw[->] (s2) edge (3,0);
\end{tikzpicture}
\caption{The book-stack, when $T[i+1] = r_i^j(I)$. This change is shown in the first two book-stacks. The third reflects the book-stack at index $i+1$ after it has been updated.} \label{fig:stack-fig}
\end{figure}

Algorithm~$\ref{algsinglewindow}$, {\tt SINGLECOUNTING}  
demonstrates co-occurrence count for a specific window length. 
The co-occurrences of an itemset can be calculated for multiple window lengths by repeating Algorithm~$\ref{algsinglewindow}$ and varying the argument $x$. 

\begin{algorithm}\SetAlgoLined
\DontPrintSemicolon
\LinesNumbered
\KwIn{Sequence $T$ of length $n$, Itemset $I$, Window Length $x$}
\KwResult{$\co(T,I,x)$}
count$\leftarrow$0\;
$\mathcal{S} \leftarrow$empty book-stack\;
\For{each item e $\in$ I}{
    $\mathcal{S}$ += (e, $-\infty$)\;
}
\For{$i  = 0 \text{ to } n-1$}{
    \If{$T[i] \in I$}{
        j $\leftarrow \mathcal{S}.${\tt find}($T[i]$)\;
        $\mathcal{S}$.{\tt update}(j)\;
    }
    \If{$i\geq x-1$ and i - $\mathcal{S}$.{\tt retrieve}(|I|)$<x$}{
        count$\leftarrow$count+1\;
    }
}
\KwRet{count\;}
\caption{SINGLECOUNTING}
\label{algsinglewindow}
\end{algorithm}

\begin{example}

Consider the sequence $T=abcabe$ and the itemset $I=\{a,b\}$. 
The algorithm initializes the $\mathcal{S}$ by adding $(e,-\infty)$ for each item $e$ in $I$, representing that element $e$ has never been seen. Table~\ref{tbl:stack}  shows the state of $\mathcal{S}$ and the resultant max \tesla{} value every time an element of $T$ is processed. At any step the max \tesla{} value can be found by taking the current index in the sequence and subtracting the last access time of the item in the bottom of the book-stack.

\begin{table*}[t!]
\caption{Book Stack changes for single-window co-occurrence counting.}
\begin{center}
\resizebox{0.75\textwidth}{!}{%
\begin{tabular}{c|c|c|c|c|c|c}
    initial&a (i=0) &b (i=1)&c (i=2)&a (i=3)&b (i=4)&e (i=5)\\\hline
    $\max\tesla$&$0-(-\infty)=\infty$&$1-0=1$&$2-0=2$&$3-1=2$&$4-3=1$&$5-3=2$\\ \hline
    \begin{tikzpicture}[stack/.style={rectangle split, rectangle split parts=#1,draw, anchor=center}]
    \node[stack=2]  {
    \nodepart{one}$a(-\infty)$
    \nodepart{two}$b(-\infty)$
};
\end{tikzpicture}
&
    \begin{tikzpicture}[stack/.style={rectangle split, rectangle split parts=#1,draw, anchor=center}]
    \node[stack=2]  {
    \nodepart{one}$a(0)$
    \nodepart{two}$b(-\infty)$
};
\end{tikzpicture}
&
    \begin{tikzpicture}[stack/.style={rectangle split, rectangle split parts=#1,draw, anchor=center}]
    \node[stack=2]  {
    \nodepart{one}$b(1)$
    \nodepart{two}$a(0)$
};  
\end{tikzpicture}
&
    \begin{tikzpicture}[stack/.style={rectangle split, rectangle split parts=#1,draw, anchor=center}]
    \node[stack=2]  {
    \nodepart{one}$b(1)$
    \nodepart{two}$a(0)$
};  
\end{tikzpicture}
&
    \begin{tikzpicture}[stack/.style={rectangle split, rectangle split parts=#1,draw, anchor=center}]
    \node[stack=2]  {
    \nodepart{one}$a(3)$
    \nodepart{two}$b(1)$
};
\end{tikzpicture}
&
    \begin{tikzpicture}[stack/.style={rectangle split, rectangle split parts=#1,draw, anchor=center}]
    \node[stack=2]  {
    \nodepart{one}$b(4)$
    \nodepart{two}$a(3)$
};  
\end{tikzpicture}
&
    \begin{tikzpicture}[stack/.style={rectangle split, rectangle split parts=#1,draw, anchor=center}]
    \node[stack=2]  {
    \nodepart{one}$b(4)$
    \nodepart{two}$a(3)$
};  
\end{tikzpicture}
\label{tbl:stack}
\end{tabular}
}
\end{center}
\end{table*}

\end{example}

\eat{I always found the use of "stack" a bit confusing. Here is a reviewer's comment from ISAAC: Time complexity analysis seems incorrect. The following statement "by keeping a hash table from each item to its node in the stack, items can be moved to the top of the stack in constant time" (Line 182-183) is not precise. The optimality of the time complexity analysis in Section 3.2 is questionable. The term stack is usually used for a first-in-last-out data structure. However, this paper does not use this definition very precisely. The stack is operated by popping out the top item, but the bottom is also referred for calculating the tesla. I do not quite get how this is implemented in O(1).}
\subsection{Complexity Analysis}
\label{sec:swl}
The book-stack can be implemented as a doubly linked list of items.  Finding elements on the bottom of the book-stack can then
be done in constant time.  We can maintain a hash table from each element to the corresponding node in the book-stack. Each node can be accessed in constant time. The book-stack will only take $|I|$ space and no additional space is needed, thus the total space is $O(|I|)$.  In addition, each element of the sequence is accessed once, and only constant time operations are performed, giving a time complexity of $O(n)$.  For co-occurrence of a single window length, this algorithm performs optimally with respect to time complexity. This is because there is an intrinsic linear cost in computing co-occurrence, as each element in the sequence must be examined in the worst case.  In the next section we present a solution that in linear time can calculate the co-occurrence for all window lengths.  

\section{All Window-Length Co-occurrence}
\label{sec:all-w-alg}

\subsection{Counting Co-occurring Windows}
To find the co-occurrence of an itemset $I=\{e_1, e_2, \ldots, e_{|I|}\}$ in sequence $T$ with window length $x$ we must count how many $x$-length windows in $T$ contain $I$.  We will make use of the fact that counting the windows containing $I$ is equivalent to counting the windows that do not contain $I$, since we know the total number of 
$x$-length 
windows is $n-x+1$.  For a sequence $T$ of length $n$ and an itemset $\{e_1\}$ we denote the $x$-length windows that do not contain $e_1$ as $\overline{\{e_1\}}_x$.  
For larger itemsets we extend the notation analogously where $\overline{\{e_1, e_2\}}_x$ are the $x$-length windows that do not contain $e_1$ and do not contain $e_2$.  A window is a {\em non co-occurrent} window as long as there is at least one element in $I$ that is not contained in the window.
Therefore, the co-occurrence of $I$ is the total number of $x$-length windows minus the number of $x$-length windows that do not contain at least one item of $I$.  
\begin{equation}
\co\left(T, I, x \right) = (n-x+1)-|\overline{\{e_1\}}_x\cup \ldots \cup \overline{\{e_{|I|}\}}_x|
\end{equation}
Using the inclusion-exclusion principle we can rewrite the co-occurrence as follows. 
\begin{equation}
    \co\left(T, I, x \right) = (n-x+1)-\sum_{\substack{A\subseteq I:\\  A\neq \emptyset}}(-1)^{|A|+1}\,|\overline{A}_x|\label{incexceq}
\end{equation}
\eat{
\eat{For a set of elements $G=\{e_1,e_2,...\}$ we define a maximal $G$-gap as a sub-string of a sequence that doesn't contain any elements of $G$ but where each of its super-strings contains an element from $G$. }
\begin{example}
\label{ex:incexc}
Consider the example of taking the co-occurrence of $\{a,b\}$ with window length $3$ in the sequence of Figure~\ref{gapfigure}. Note that by Equation~\ref{incexceq}, $\co(T,\{a,b\},3) = (11-3+1)-(\overline{\{a\}}_3+\overline{\{b\}}_3-\overline{\{a,b\}}_3)=9-(1+3-0)=5$.  
\end{example}
}
\begin{figure}[t]
    \centering
    \includegraphics[width=0.33\textwidth]{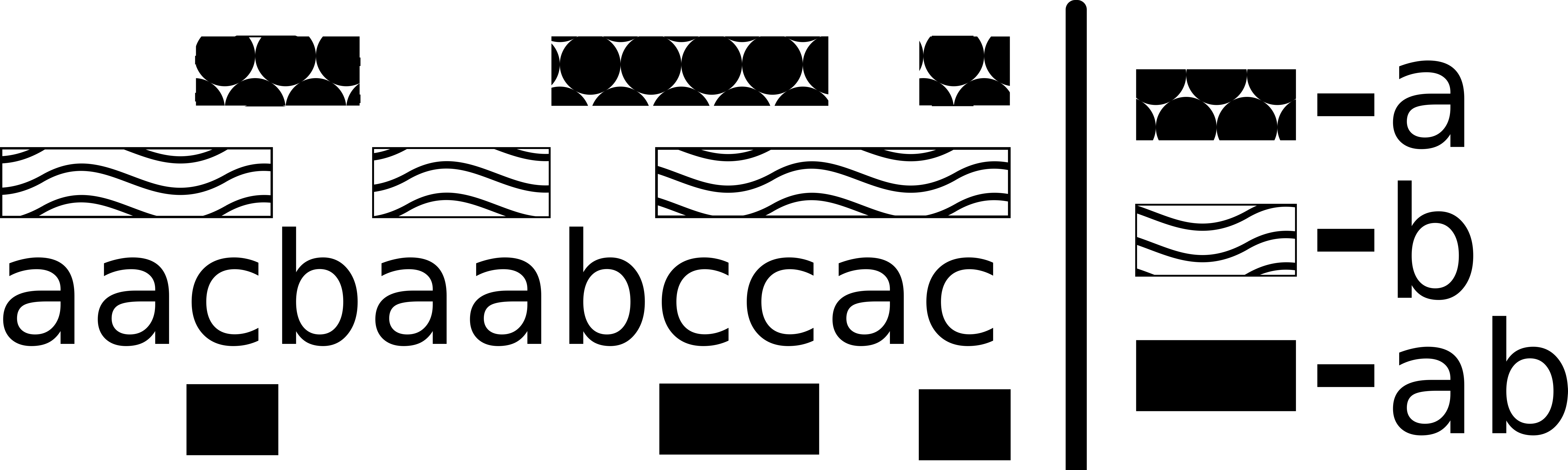}
    \caption{Gaps for certain elements in a sequence.  The uppermost pattern illustrates the three gaps `a'-gaps, the middle pattern shows the `b'-gaps, and the bottom pattern shows the three gaps that contain neither `a' nor `b'.}
    \label{gapfigure}
\end{figure}
We know that an $A$-gap of size $k$ contains $k - x + 1$ windows of length $x$ in which none of $A$ occurs. Thus, $|\overline{A}_{x}| = \sum_{k = x}^{n}(k-x+1)N_{A}(k)$, where $N_{A}(k)$ is the number of $A$-gaps of length $k$. Working with the right term of equation (\ref{incexceq}),
\begin{align}
    \sum_{\substack{A\subseteq I:\\  A\neq \emptyset}}(-1)^{|A|+1}\,|\overline{A}_x| 
    &= \sum_{\substack{A\subseteq I:\\  A\neq \emptyset}}(-1)^{|A|+1}\,\sum_{k = x}^{n}(k-x+1)N_{A}(k) \\
    &= \sum_{k = x}^{n}(k-x+1)\sum_{\substack{A\subseteq I:\\  A\neq \emptyset}}(-1)^{|A|+1}N_{A}(k).
\end{align}
Now, let us define:  
\begin{equation}
\label{eq:hist}
H[k] = \sum_{\substack{A\subseteq I:\\  A\neq \emptyset}}(-1)^{|A|+1}N_{A}(k)
\end{equation}
We call the collection of $H[k]$'s for all values of $k$ a {\em gap histogram}, $H$. 
The co-occurrence of $I$ in $x$-length windows of sequence $T$ is then calculated as follows.  
\begin{equation}
\co\left(T, I, x \right) = (n-x+1)-
\sum_{k=x}^{n}(k-x+1)H[k]
\label{eq:histwindow}
\end{equation}
\eat{It further follows from equation \ref{incexceq} and theorem \ref{thm:gapcount} that the co-occurrence of $I$ can be computed by taking the number of length $x$ windows that are a sub-string of a gap of an odd cardinality subset of $I$ and subtracting the number of length $x$ windows that are a sub-string of a gap of an even cardinality subset of $I$.  We will consider a histogram $H$ where $H[i]$ is equal to the number of maximal gaps of length $i$ of odd subsets of $I$ minus the number of maximal gaps of length $i$ of even subsets of $I$.  It clearly follows then that to compute the co-occurrence with a window length $x$ we can compute it with the formula $(n-x+1)-\sum\limits_{i=x}^n{(i-x+1)H[i]}$.  }
An elegant property of this equation is that by 
storing the cumulative counts in a gap histogram we can simultaneously calculate the co-occurrence for all window lengths.  
\eat{
\begin{example}
Figure~\ref{histogram} shows the gap histogram generated from the sequence in Figure~\ref{gapfigure}.  We can see that $H[1] = -1$. This is because there are two length one 
$\{a,b\}$-gaps and one length one  $\{b\}$-gap, thus there is one odd sized length one-gap and two even sized length one-gaps which gives us  $H[1]=1-2=-1$.  For $H[2]$ we have one length two  $\{a\}-,\{b\}-, \text{and } \{a,b\}-$ gaps for a net contribution of $2-1=1$.  For $H[3]$ we have one length three $\{a\}$- gap and one length three $\{b\}$-gap, so $H[3]=2$.  Finally, we have that for $H[4]$ there is a single length 4 $\{b\}$ gap so $H[4]=1$.
\end{example}
\begin{wrapfigure}{i}{0.3\textwidth}
    \centering
    \includegraphics[width=0.3\columnwidth]{figs/hist_values.png}
    \caption{Gap histogram for the sequence of  Figure~\ref{gapfigure} and itemset \{a, b\}.}
    \label{histogram}
\end{wrapfigure}
}
Using gap histograms to store cumulative counts has a space complexity of $\sqrt{n|I|}$. 
In Theorem~\ref{theo:spacecomp}, we will formally discuss the space complexity in more 
details. 
Calculating the co-occurrence from the gap histogram instead of directly counting co-occurrent windows is beneficial since calculating gaps does not require a window length as input and yet the gap information is still sufficient to easily calculate the co-occurrence for all window lengths. 
Thus, all that is needed to calculate all window length co-occurrence is an algorithm to generate the gap histogram. 

The simplest way to generate the gap histogram is to iterate through the sequence, keeping track of where gaps begin and end.  Whenever an item in the given itemset is found at some index $i$ it marks the end of a gap for any subset of $I$ containing that item and also marks the beginning of a new gap spanning from $T[i+1 ... k-1]$, where $k$ is either the index of the next occurrence of an element in the subset of $I$ in the sequence or $n$ if another element does not occur before the end of the sequence. Note that if an element in $I$ occurs in two adjacent indices in the sequence ($i$ and  $i+1$), we obtain the gap $[i+1, i]$ which we treat as a length $0$ gap and discard.  The length $l$ of each newly ended gap can be updated in the histogram by either incrementing or decrementing $H[l]$ depending on whether the subset size was odd or even respectively.  The pseudocode of this approach can be found in the appendix under, Algorithm~\ref{algnaive}, {\tt GAPCOUNTING}.  This algorithm has run time $O(n2^{|I|})$ and performs poorly for large itemsets.  
Algorithm~\ref{algnaive} is inefficient since whenever an item from $I$ is encountered in the sequence, we need to consider $2^{|I|-1}$ subsets of $I$ and update the histogram (subtract or add counts) accordingly (Line~\ref{line:subset}). 
A better algorithm is presented next. 

\newcommand{\pluseq}{\mathrel{{+}{=}}}

\newcommand{\dblop}[2]{%
  \mathrel{#1}\mathrel{#2}%
}

\subsection{Efficient Gap Counting}

Since  updates to the histogram have negating effects on each other (Equation~\ref{eq:hist}), many of the histogram entries do not change when an item of the itemset is observed in the sequence.  
 It turns out when an item of the itemset is observed in a sequence, we only need to update the histogram for the gaps related to the {\em first and second least recently seen items} of $I$. 
To keep track of the $\tesla$'s, as we iterate through the sequence we can maintain a book-stack data structure that contains each item in $I$ along with the time that it last appeared in the sequence, so that the most recently seen item appears at the top of book-stack.

Observe that, when an item $e$ from $I$ is seen in the sequence at index $i$, a maximal gap representing each subset of $I$ containing $e$ is added to the histogram.  Furthermore, for any one of those sets $G$, the length of the added gap is the minimum tesla value attained by an item in $G$ at index $i-1$.  Note that in this context we take $\text{tesla}(e,i)=i$ if the element has not yet been encountered in the sequence.  These gaps account for all of the gaps in the sequence except for gaps that include the final element of the sequence, these gaps are handled specially.

For the following theorem, we first provide some notation.
Let
\[
H_{i}=(H_{i}[1],H_{i}[2],\ldots,H_{i}[n])
\]
be the histogram up to index $i$ in the sequence.
\begin{theorem}\label{thm:hupdates}
    For any $0<i<n$,
    suppose $T[i] = r_i^{|I|}(I)$ then
    \[
        H_{i}[k] = 
        \begin{cases}
            H_{i-1}[k] + 1\ \text{if}\  k =  \tesla( r_i^{|I|}(I) ,i-1) \\
            H_{i-1}[k] -1\ \text{if}\  k = \tesla( r_i^{|I|-1}(I), i-1) \\
            H_{i-1}[k]\  \text{otherwise.}
        \end{cases}
    \]
    If $T[i] \neq r_i^{|I|}(I)$ then $H_{i}[k] = H_{i-1}[k]$ for all $k$.
    In other words, the histogram is only updated when the next element in the sequence is the item that was just at the bottom of the book-stack.
\end{theorem}
\begin{proofsketch}
    We have a maximal gap for every subset $A$ of $I$ containing $T[i]$. The length of this $A$-gap is $\min_{e \in A}\tesla(e,i-1)$, hence the addition to the histogram from $A$ is $(-1)^{|A|+1}$ to the $k$th spot where $k = \min_{e \in A}\tesla(e,i-1)$. Suppose $T[i] \neq \arg\max_{e\in I}\tesla(e,i-1)$ i.e., $T[i]$ is not the item seen furthest in the past most recently. There are the same number of even and odd subsets of $I$ in which $T[i] = \arg\min_{e \in A} \tesla(e, i-1)$ hence these subsets contribute no net updates to $H$. For the remaining subsets, the same argument follows, hence there are no net updates.
    
    Now suppose that $T[i] = \arg\max_{e \in I} \tesla(e, i-1)$. Similar to the above, for each item in $I$ not equal to $r_i^{|I|-1}(I)$ and $T[i]$, there are the same number of even and odd subsets of $I$ in which $T[i] = \arg\min_{e \in A} \tesla(e,i-1)$. But for $r_i^{|I|-1}(I)$ there is but one subset in which this is satisfied, namely, $\{r_i^{|I|-1}(I), T[i]\}$, and there is also one subset in which $T[i]$ satisfies this, $\{T[i]\}$. Therefore we have 
    \begin{align*}
        H_{i}[\tesla(T[i],i-1)] &= H_{i-1}[\tesla(T[i],i-1)] + 1,\\
        H_{i}[\tesla(r_{i}^{|I|-1}(I),i-1)] &= H_{i-1}[\tesla(r_{i}^{|I|}(I),i-1)] - 1.
    \end{align*}
\end{proofsketch}
The theorem does not handle the case for $H_{n}$, which we now address. 
The argument is similar to the proof above for $H_i$ with $i<n$, except that $T[i]$ is undefined.
All gaps necessarily close at the end of the sequence. This means that $|C_{k}| = \sum_{\ell = 0}^{|I|}\binom{|I| - j}{\ell}$, for all but $j = |I|$. For $j = |I|$ there is but one set for which $r_{i}^{|I|}(I) = \arg\min_{e \in A}\tesla(e,n)$, namely, $\{r_{i}^{|I|}(I)\}$. Thus $H_{n}[k] = H_{n-1}[k]$ for all $k$ except when $k = \tesla(n, i-1)$ in which $H_{n}[k] = H_{n-1}[k] - 1$. 
\eat{
, and the set $B$
simply consists of $I- \{ e_{\ell}, \ldots, e_1 \}$.
The set $B$ is empty only
when $\ell = |I|$, that is, we are examining the bottom element of the book-stack. Therefore
there is a non-zero update $\delta(n,k)$ only 
for $k = \mathrm{tesla}(e_{|I|},n)$.
In this case, $A = \{e_{|I|}\}$, and $\delta(n,k) = 1$. 
}

The incremental updates that we have derived above result in algorithm \name{}, shown in Algorithm~\ref{algfinal}. 

\subsection{Complexity Analysis}
\label{sec:complexity}
The next two theorems assume that the histogram can be implemented as a hashtable with perfect hashing.  Without perfect hashing the histogram must contain space for all entries from $1-n$ and thus will be linear in space to maintain a constant run time or constant histogram updates must be sacrificed to obtain a worst case $n^2$ runtime.
\begin{theorem}[Time Complexity]
\label{theo:timecomp} The time complexity of all-window length co-occurrence algorithm is linear in the length of the sequence.  
\end{theorem}
\begin{proof}
The algorithm iterates over the sequence once and possibly updates the book-stack and the histogram for each element in the sequence.    
Since updating the book-stack and updating the histogram are both done in constant time, the generation of the histogram is done in linear time in the length of the sequence. 
Once a histogram is computed, the co-occurrence for every window length is computed in a linear time by summing the histogram as shown in Equation~\ref{eq:histwindow}.  Thus, the algorithm provides an $O(n)$ method to calculate all window length co-occurrence.
\end{proof}

\begin{theorem}[Space Complexity]The space complexity of the algorithm is $O(\sqrt{n|I|})$ where $n$ is the length of the sequence and $|I|$ is the size of the itemset.
\label{theo:spacecomp}
\end{theorem}
\begin{proof}
    Space is used to maintain the book-stack and the histogram.  The book-stack will use $O(|I|)$ space.  Note that for any item $e$ in the itemset the total length of gaps for $\{e\}$ is at most the length of the sequence.  Thus, we have that the sum of all of the lengths of single-item gaps is bounded above by $n|I|$.  Furthermore, whenever an item of the itemset is on the bottom of the book-stack a maximum of two new  gaps are added to the histogram.  The length of the gap associated to the bottom item in the book-stack is equal to the length of a single-item gap. The length of the other gap is bounded above by the length of the first gap. 
    Therefore, the sum of the length of all gaps added to the histogram is bounded above by $2n|I|$.  Note the size of the histogram is the number of distinct gap lengths added to it. In the worst case, gaps are greedily added to the histogram such that there is a length $1,2,\ldots,k$ size gap added.  In this case, if the total number of gaps added is $k$ the total length of the gaps is $\frac{k(k+1)}{2}$. We know that the sum of the gaps length in a histogram is bounded by $2n|I|$. Thus, we have that $\frac{k(k+1)}{2}\leq 2n|I|$. Solving  for $k$, we have that $k^2+k\leq 4n|I|$ and $k\leq 2\sqrt{n|I|}$.  Thus, the total space used is bounded above by $|I|+2\sqrt{n|I|}$ which gives a space complexity of $O(\sqrt{n|I|})$.  
\end{proof}

\begin{algorithm}\SetAlgoLined
\DontPrintSemicolon
\LinesNumbered
\KwIn{Sequence $T$, ItemSet $I$}
\KwResult{Co-occurrence of all window lengths}
    $H \leftarrow$ empty histogram\;
    $cooc \leftarrow$ []\;
    $\mathcal{S}\leftarrow$ empty book-stack\;
    \For{$e \in I$}{
        $\mathcal{S}$ += (e, $-\infty$)\;
    }
    
    \tcp*[h]{Read through entire sequence}\\
    \For{$i=0$ to $n-1$}{
        current$\leftarrow T[i]$\;
        \tcp*[h]{When element is seen, update bottom two gaps}\;
        \If{$\text{current} \in I$}{
            \If{$\mathcal{S}.{\tt find}(\text{current}) = |I|$}{
                $f \gets i - \mathcal{S}.{\tt retrieve}(|I|)$\;
                $s \gets i - \mathcal{S}.{\tt retrieve}(|I|-1)$\;
                $H[f] \gets H[f] + 1$\;
                $H[s] \gets H[s] - 1$\;
            }
            j $\leftarrow \mathcal{S}.{\tt find}(\text{current})$\;
        $\mathcal{S}.{\tt update}(j)$\;
        }
        }
        \tcp*[h]{Final gap from bottom of book-stack}\\
        $f \gets i - \mathcal{S}.{\tt retrieve}(|I|)$\;
        $H[f] \gets H[f]+1$\;
        \For{$x = |I|$ to $|T|-|I|+1$}{
            $S_x \gets 0$\;
            \For{$k = x$ to $|T|$}{
                $S_x \leftarrow S_x + (k-x+1)H[k]$\;
            }
            $cooc[x] \leftarrow (|T| - x + 1) - S_x$\;
        }
    \Return{$cooc$}\;
    \caption{\name{}}
    \label{algfinal}
    \end{algorithm}

\section{Pattern Co-occurrence}
We now wish to generalize our algorithm in two ways. The first is to patterns and the second is to a stream in which multiple events can occur at the same index. Pattern co-occurrence is explained first. A pattern is a string with characters drawn from our alphabet $\mathbb{A}$. A pattern $e$'s $i $th component is denoted $e[j]$ and the length of the pattern is $|e|$.  A pattern occurs in a sequence $T$ if there exists $j \in \left[ |T| \right]$ such that $T[j...j+|e|-1] = e$, also let all such $j$ be denoted in the set $b(e)$. Thus, pattern co-occurrence for an itemset $I$ is defined as the number of windows in which every pattern in $I$ occurs. 
We wish to find an algorithm that can compute the co-occurrence for all window lengths in one pass for patterns. It is clear that $\tesla$ is no longer well-defined. Let $e$ be a pattern. So define $\btesla(e,i) = i - \max (|b(e)\cap\{0,\ldots ,i\}|)$, which is the distance between $i$ and the most recent start of the pattern. If $b(e) \cap \{0,\ldots,i\}$ is empty, then let it be $i$.

We can use our previous definition of an $A$-gap for $A\subseteq I$, but the size of an $A$ gap is now found differently. Previously, the size of an $A$-gap closed at time $i$ would be $\min_{e \in A}\tesla(e, i-1)$, but now it is $\min_{e\in A}\btesla(e,i-1)$, since an $A$-gap still occurs if all but the tail ends of members of $A$ are within said gap. Supposing that no two patterns in consideration end at the same time, it is easy to see that Theorem \ref{thm:hupdates} still holds in this case, using $\btesla$ in place of $\tesla$. Thus finding an algorithm that handles multiple events at the same index would solve all window length pattern co-occurrence as well.

\subsection{Multiple Item Co-occurrence}
It is now natural to define co-occurrence for sets of items. We let $T[i]\subseteq \mathbb{A}$, rather than just one element of $\mathbb{A}$, for all $i$. A co-occurring window for some itemset $I\subseteq \mathbb{A}$ is a window in which for all $e \in I$, there exists a set $A \in w(x,i)$ such that $e \in A$. Thus the co-occurrence is the sum of these co-occurring windows. This is the natural extension. We will now present the following theorem relating to the updates of $H$. Let $X_{i}$ denote the set of items that occur at $T[i]$.

\begin{figure}
\centering
\begin{tikzpicture}[>=stealth,thick,baseline]
    \matrix [matrix of math nodes, nodes={align=center,
    text width=1.5cm}] {
    T[i-k_1] & & T[i-k_2] &  &  & T[i] \\\hline
    x_1      & x_2 & x_3      &  y_3 & x_4      & x_1 \\
             &     & y_1      &      &  & x_2 \\
             &     & y_2      &      &  & x_3 \\
             &     &          &      &  & x_4 \\
             &     &          &      &  & x_5 \\
             };
    \end{tikzpicture}
\caption{Illustration of Theorem~\ref{thm:mult}.  
The set of patterns $I$ consists of $x_1,x_2,\ldots$, which 
were seen at $T[i]$, and all other patterns $y_1,y_2,\ldots$.
Here $A$ is the set $\{x_1,x_2\}$ of patterns seen at $T[i]$ that 
were last seen further in the past than any of the other patterns $y_1,y_2,\ldots$.
We add one to $H[k_1]$, where $k_1$ is the time elapsed since $x_1$ was last
seen.  We subtract one from $H[k_2]$, where $k_2$ is the time elapsed since $y_1$ was last seen.}
             \label{fig:pawlco}
\end{figure}

\begin{theorem}\label{thm:mult}
    Suppose that for all $a \in I\cap X_{i}$, $\tesla(a,i-1) < \tesla(e,i-1)$ for any $e \in I\setminus X_i$. Then for any $1 < i < n$,
    \[
        H_{i}\left[ k \right] =
        \begin{cases}
            H_{i-1}\left[ k \right] + 1\ \text{for}\ k = \tesla(r_{i}^{|I\setminus X_i|}(I\setminus X_i),i-1)\\
            H_{i-1}\left[ k \right] - 1\ \text{for}\ k = \tesla(r_{i}^{|I|}(I),i-1)\\
            H_{i-1}\left[ k \right]\ \text{otherwise.}
        \end{cases}
    \] 
    Otherwise, $H_{i}\left[ k \right]= H_{i-1}\left[ k \right]$ for all $k$.
\end{theorem}
\begin{remark}
    This is a generalization of Theorem \ref{thm:hupdates}. Observe that in the case when $|A| = 1$, this reduces to that result. Moreover, when $i= n$ the same update follows.
\end{remark}
\begin{proofsketch} The incremental updates to $H$ correspond to all the subsets of $I$ that contain at least one member of $X_{i}$. Weakly order $X_{i}$ according to $\tesla$. Now let $A_{j}$ where $j\in [X_{i}]$, be the set of subsets of $I$ that contains $x_{j}$. Let $\left< A \right>_{i}$ be the updates corresponding to some set $A$. Therefore
    \[
        H_{i} - H_{i-1} = \sum_{J\subseteq \left[ |K| \right]}(-1)^{|J|+1}\left< \bigcap_{j \in J} A_{j} \right>.
    \] 
    Consider each $\mathcal{K}_{J} = \bigcap_{i \in J}A_{i}$. For each one, the update is the same if one removes all members of $X_{i}$ besides the one corresponding to the smallest number in $J$, call this set $\mathcal{K}'_{J}$. Using Theorem \ref{thm:hupdates}, the update is $+1$ for $k = \tesla(x,i-1)$, $x$ being the item described before, and also is $-1$ for $k = \tesla(x,i-1)$ for $x$ being the furthest item seen in the past not in $\mathcal{K}_{J}\cap I$. But this implies that $J$ that are not of the form $J_{m} = \{|X_{i}| - m ,\ldots , |X_{i}|\}$. For any $0 \leq m < |X_{i}|$, the positive update corresponding to $\mathcal{K}_{J_{m}}$ cancels with the negative update corresponding to $\mathcal{K}_{J_{m+1}}$. This process telescopes leaving only the postive update corresponding to $\mathcal{K}_{J_{0}}$ and the negative update corresponding to $\mathcal{K}_{J_{|X_{i}|}}$. This gives the desired result.
\end{proofsketch}

A full proof is given in the appendix. Figure~\ref{fig:pawlco} provides an illustration of Theorem~\ref{thm:mult}. With this result we can now construct a similar algorithm to those before, with a few modifications. Maintain a book-stack as before, but notice that it is no longer a strict ordering. For example, if $X_i = \{e_1,e_2\}\subseteq I$, then one of $e_1$ and $e_2$ will occupy the top of the book-stack and the other will occupy the the second to top spot. To check whether $\max_{e \in X_i \cap I}\tesla(e,i) < \min_{e \in I \setminus X_i}\tesla(e,i)$, we partition the book-stack using $p \in \{0,...,|I|\}$, where $p$ is defined as follows: for all $j \leq p$, $\tesla(r_i^j(I),i-1) = \tesla(r_i^{|I|}(I),i-1)$, and for all $j > p$, $\tesla(r_i^j(I),i-1) > \tesla(r_i^{|I|}(I),i-1)$. Thus checking if the non-trivial conditions given in Theorem~\ref{thm:mult} hold is easy as we just check that $r_i^j(I) \leq p$ for every $j$ corresponding to a member in $X_i$. It is also easy to update the histogram if these conditions hold, as we just update according to $r_j^p(I)$ and $r_j^{p+1}(I)$. The pseudocode is given in algorithm~\ref{algmult} in the appendix.
\subsection{Complexity Analysis}
Maintaining the partition is at the worst case a linear scan of $I$ at each index. A state machine is spawned at each element that starts one of the patterns, and is terminated either by the pattern not being completed, or by completion of the pattern. If completed, that pattern is moved from its current level in the structure to the level corresponding to the now found value of $\btesla$. This requires $O(|I|)$ operations. Additionally, there can be at most $2|I|$ state machines created or terminated at each step. Thus, the time complexity is $O(n|I|)$. Space complexity is also the same but including the space for the state machines, giving $O(\sqrt{n|I|} + e_{\max}|I|)$, where $e_{\max}$ is the length of the largest pattern.

\section*{Continuous co-occurrence}
The previous section opens up new opportunities. Because we can now allow multiple items to occur at the same index, we can analyze occurrences of discrete events, which are items in the orignal sense, in an interval of time, which is the sequence in the original sense. Given a set of events and times, we would like to approximate the probability of two events occuring within some interval of time with generalized AWLCO\@.

To be more precise, let $\tau \geq 0$ and define 
\[
    T = \{(e_1,\omega_{1,n_1}) ,\ldots ,(e_1,\omega_{1,n_{k_1}}), (e_{2},\omega_{2, 1}) ,\ldots , (e_{|I|}, \omega_{|I|,n_{|I|}})\}
\]
in which $e_{i} \in I$ and $\omega_{i,j}\in \left[ 0,\tau \right]$ for all $j \in \{1 ,\ldots , n_{i}\}$, for all $i \in \{1 ,\ldots ,|I|\}$. We call this $T$ the set of time stamps of our set of events $I$. Let, for the sake of convenient notation,
\[
    \Pr(I \in \left[ a,b \right],T,r) = \Pr_{y \sim U(\left[ a,b \right])}( \forall e \in I, \exists (e, \omega) \in T : |y-\omega| < r).
\]
We will attempt to find a discrete analog of the above probability. Let $Q_{t}(T) = \{(e, \floor*{\frac{\omega}{t}}t) : (e,\omega) \in T\}$. $Q_{t}(T)$ is a means of disretizing the possible values of $\omega$ that could lie in $T$. Notice that $\lim_{t \to \infty}Q_{t}(T) = T$, which says that as $T$ gets smaller, $Q_{t}(T)$ becomes a better and better approximation of $T$. Now let $\Tr_{t}(T) = \{a_{0}, a_{1},\ldots,a_{n} \}$ in which $a_{i} = \{e : \exists (e,\omega) \in T, \floor*{\frac{\omega}{t}} = i\}$ for all $i \in \{0,\ldots ,n\}$. This is what we will feed into generalized AWLCO in order to approximate the above probability. It is clear that $\Tr_{t}(T) = \Tr_{t}(Q_{t}(T))$. Let $Tr_{t}(T) + x = \{a_{-x/2} ,\ldots ,a_{-1},a_{0},a_{1},\ldots ,a_{n},a_{n+1},\ldots,a_{n + x/2} \}$ in which $a_{-x/2} = a_{-x/2 + 1} = \cdots  = a_{-1} = a_{n + 1} = \cdots = a_{n + x/2} = \emptyset$, but all other $a_{i}$ are given as above. We now have the following result which allows us to translate co-occurrence into probabilities. 

\begin{theorem}
    \[
    \co(\Tr_{t}(T) + x,I,x) = \frac{1}{t}\int_{0}^{\tau}\mathbbm{1}{(I\in\left[ \max\{0,y-x/2\}, \min\{\tau, y + x/2\} \right],Q_{t}(T))}dy
\]
\end{theorem}
\begin{remark}
The right hand side is simply the measure of the set of $y \in \left[ 0,\tau \right]$ such that all of $I$ occurs within distance $x/2$ of $y$.
\end{remark}
\begin{proof}
    We clearly have that
    \begin{equation}
        \co(\Tr_{t}(T) + x,I,x) = \sum_{i = x/2}^{n + x/2}\mathbbm{1}(I \subseteq w(x,i)).
    \end{equation}
    Now if $I \subseteq w(x,i)$, that means for all $e \in I$ there exists $(e,\omega)\in T$ such that $\floor*{\frac{\omega}{t}}\in \{i-x + 1,\ldots ,i \}$. This means that $\floor*{\frac{\omega}{t}}t \in [ (i -x + 1)t,(i+1)t)$. If $I\not\subseteq w(x,i)$, then for some $e\in I$ there is no such $(e,\omega)\in T$ in which $\floor*{\frac{\omega}{t}}t \in [ (i-x + 1)t,(i+1)t )$. We can therefore write
    \begin{align}
        \sum_{i = x/2}^{n+x/2}\mathbbm{1}(I \subseteq w(x,i))
    &= \sum_{i = x/2}^{n+x/2}\mathbbm{1}(I \in \left[ (i-x + 1)t,(i+1)t \right),Q_{t}(T)) \\
    &= \int_{x/2}^{n+x/2}\mathbbm{1}(I \in \left[ (\floor{y} + 1 - x)t, (\floor{y} + 1)t \right),Q_{t}(T))dy \\
    &= \frac{1}{t}\int_{xt/2}^{(n+x/2)t}\mathbbm{1}\left(I \in \left[ \left( \floor*{\frac{y}{t}} - x +1\right)t, \left( \floor*{\frac{y}{t}} + 1\right)t  \right),Q_{t}(T)\right)dy.
\end{align}
But now recall that for every $(e,\omega)\in Q_{t}(T)$, that $\omega = kt$ for some $k \in \mathbb{Z}_{+}$. So suppose that $I \in \left[ \left( \floor*{\frac{y}{t}} - x + 1\right)t, \left(\floor*{\frac{y}{t}} + 1\right)t \right)$. Then for every $e \in I$, there exists $(e,\omega) \in Q_{t}(T)$ such that $\omega/t \in \{ \floor*{\frac{y}{t}}-x +1,\ldots , \floor*{\frac{y}{t}}\}$.
Since $y/t - x < \floor*{y/t} - x + 1$ and that $y/t \geq \floor*{y/t}$, we have that $I \in \left( y - xt, y \right]$. Now suppose that $I \not\in \left[ \left( \floor*{\frac{y}{t}} - x + 1 \right)t, \left( \floor*{\frac{y}{t}}+1\right)t \right)$ then for some $e \in I$, there is no such $(e,\omega) \in Q_{t}(T)$ in which $\omega/t \in \{ \floor*{\frac{y}{t}}-x +1,\ldots , \floor*{\frac{y}{t}}\}$. Now since $y/t -x \geq \floor{y/t} - x$ and $y/t < \floor{y/t} + 1$ we have that $I \not\in \left( y - xt, y \right]$. Therefore 
\begin{align}
    &\frac{1}{t}\int_{xt/2}^{(n+x/2)t}\mathbbm{1}\left(I \in \left( \left( \floor*{\frac{y}{t}} - x \right)t, \floor*{\frac{y}{t}}t \right],Q_{t}(T)\right)dy \\
    &=  \frac{1}{t}\int_{xt/2}^{(n+x/2)t}\mathbbm{1}\left( I \in \left( y - xt,y \right],Q_{t}(T) \right)dy \\
    &=\frac{1}{t}\int_{xt/2}^{(n+x/2)t}\mathbbm{1}\left( I \in \left[ y - xt,y \right],Q_{t}(T) \right)dy\\
    &= \frac{1}{t}\int_{0}^{nt}\mathbbm{1}\left( I \in \left[ y - xt/2,y+xt/2 \right],Q_{t}(T) \right)dy\\
    &= \frac{1}{t}\int_{0}^{\tau}\mathbbm{1}\left( I \in \left[ \max\{0,y - xt/2\},\min\{\tau,y+xt/2\} \right],Q_{t}(T) \right)dy,
\end{align}
where the last inequality follows from $nt = \tau$ and that no members of $I$ lie below $0$ or above $\tau$.
\end{proof}

The previous theorem shows us that we can represent co-occurrence as a continuous sum, which makes it much easier for us to acheive an error bound for continuous co-occurrence. Let us now find this bound. 

\begin{theorem}
    Let $I(A) =\int_{0}^{\tau}\mathbbm{1}{(I\in\left[ \max\{0,y-x/2\}, \min\{\tau, y + x/2\} \right],A)}dy$. Then if $x > t$,

    \begin{equation}
        |I(Q_{t}(T)) - I(T)| < 
        \tau \frac{t}{x}
\end{equation}
\end{theorem}
\begin{proofsketch}
    We first write
    \begin{align}
        I(Q_{t}(T)) - I(T) &= \int_{0}^{\tau}\mathbbm{1}{(I\in\left[ \max\{0,y-x/2\}, \min\{\tau, y + x/2\} \right],Q_{t}(T))}dy \\
        &-\int_{0}^{\tau}\mathbbm{1}{(I\in\left[ \max\{0,y-x/2\}, \min\{\tau, y + x/2\} \right],T)}dy,
    \end{align}
    so it suffices to bound the length of region where the integrands differ. For the sake of ease, denote
    \[
        \mathbbm{1}(y,Q_{t}(T)) = \mathbbm{1}{(I\in\left[ \max\{0,y-x/2\}, \min\{\tau, y + x/2\} \right],Q_{t}(T))}
    \] 
    and
    \[
        \mathbbm{1}(y,T) = \mathbbm{1}{(I\in\left[ \max\{0,y-x/2\}, \min\{\tau, y + x/2\} \right],T)}.
    \]

    The trick here is to consider only $y \in \left[ kt,(k+1)t \right]$ at a time. It can be shown that if $t > x/2$, then the measure of the set in each of these segments in which $\mathbbm{1}(y,T)$ and $\mathbbm{1}(y,Q_{t}(T))$ differ, is at most $t-x/2$. There are at most $\tau/t$ such segments. Therefore in this case $|I(Q_{t}(T)) - I(T)| < \frac{\tau}{t}(t- x/2) = \tau - \tau\frac{x/2}{t}$ And since $x/t > 1/2$, we have that $\tau -\tau\frac{x/2}{t} < \tau/2 < \tau\frac{t}{x}$.

    Now suppose that $t \leq x/ 2$. It can be shown now that measure of the set of $y$ in each segment that differ are at most $t$. But if this were to happen, it would imply that they do not differ for the following $x/2t$ segments on either side. This means there is error at most $t$ for $x/t$ segments each of length $t$ which gives a total error bound of $\tau\frac{t}{tx/t} = \tau \frac{t}{x}$.
\end{proofsketch}

Now putting these results together, we obtain that
\begin{equation}
    |\co(\Tr_{t}(T) + x, I ,x)t/\tau - \Pr (I \in \left[ a,b \right], T ,x/2)| < \frac{t}{x}.
\end{equation}

\section{Related Work} 
\label{sec:relatedwork}

{\bf Counting in Streams - } In count-distinct problem, the goal is to know the number of unique elements in a stream~\cite{DBLP:conf/focs/FlajoletM83,MankuM02}. 
In bit-counting problem, the goal is to maintain the frequency count of 1's in the last $k$ bits of a bit stream of size $N$. 
Datar et al.\ propose an approximate algorithm with for the bit-counting problem with $O(\log^2k)$ space complexity~\cite{DBLP:journals/siamcomp/DatarGIM02}. 
Existing counting algorithms for streams assume the {\em sliding-window} model of computation, that is answering queries or mining is done over the last $w$ elements seen so far~\cite{DBLP:books/sp/16/DatarM16}.  
However, \name{} introduces a new analysis model -- all-window-length analysis model -- which is compelled to analyze and query all windows of all lengths starting from the beginning of a stream or anytime in the the past. 
To that end, \name{} presents an efficient and exact itemset counting algorithm for the all-window-length analysis model.   

The frequent itemset mining in stream is a well-studied problem that adheres to the counting problem~\cite{CormodeH08}.  
The seminal work by Manku and Motwani presents an algorithm for estimating the frequency count of itemsets in a stream and identify those itemsets that occur in at least a fraction $\theta$ of the stream seen so far with some error parameter $\epsilon$~\cite{MankuM02}. 
For example, when the input is a stream of transactions where each transaction is a set of items, the goal is to find the most frequent itemsets within transactions. The challenge is to consider variable-length itemsets and avoid the combinatorial enumeration of all possible itemsets. 
Many existing frequent itemset mining algorithms 
(with exception of~\cite{DBLP:conf/icdm/LeungK06,DBLP:conf/kdd/ChangL03}) obtain approximate results with error bounds. 
A variation of frequent itemset mining is the  problem of mining frequent co-occurrence patterns across multiple data streams~\cite{YuY0LP15}. 
The definition of co-occurrence patterns is slightly different than co-occurrence itemsets considered by \name{}. 
A co-occurrence pattern is a group of items that appear consecutively showing tight correlations between these items. 
A frequent co-occurrence pattern is the pattern that appears in at least $\theta$ streams within 
a time period of length $\tau$ 
and the appearance of the pattern in each stream  happens within a time window of $\delta$ or smaller. 
In this paper, \name{} presents an all-window  length frequency counting for a query itemset. 
A natural extension of the itemset frequency counting of presented by \name{} is mining 
frequent itemsets in all window-lengths. 

{\bf Affinity Analysis - } 
Zhong et al.\ defined \emph{reference affinity} for data elements on an access trace. A set of data elements belong to the same affinity group if they are always accessed close to each other~\cite{Zhong+:PLDI04}.  The closeness is defined by $\mathit{k-}$linked-ness.  They proved that reference affinity forms a unique partition of data for every $k$, and the relation between different $k$s is hierarchical, i.e.\ the affinity groups at link length $k$ are a finer partition of the groups at $k+1$.  This definition requires \emph{strict} co-occurrence in that every occurrence of a group element must be accompanied by all other elements of the group.  \emph{Weak reference affinity}~\cite{Zhang+:findingreference04} introduces a second parameter, affinity threshold.  It adheres to the unique and hierarchical partition properties with respect to both parameters.  
Zhang et al.\ showed that neither strict reference affinity, nor weak reference affinity can efficiently be computed~\cite{Zhang+:POPL06}. Thus they gave a heuristic solution and adapted it to use sampling.  The average time complexity of their algorithm is $O(N\delta\omega^2+N\delta\pi)$, where $N$ is the length of the trace, $\delta$ is the sampling rate, $\omega$ is the size of the affinity group, and $\pi$ is the average time length of windows containing accesses to all members of the group $\omega$.  Lavaee et al.\ gave an $O(L\delta\omega^2)$ algorithm to compute the affinity for all sub-groups of sizes up to $\omega$~\cite{Lavaee+:CC19}.  Reference affinity has been used to optimize the memory layout in data structure splitting~\cite{Zhong+:PLDI04}, whole-program code layout~\cite{Lavaee+:CC19}, and both~\cite{Zhang+:POPL06}. 


\section{Discussion and Future Work}

{\bf Applications} The all-window-length co-occurrence has applications in text analysis, the optimization of the memory layout of programs, and accelerating the search for RNA sequences in genomes. 
In terms of practical applications, our plan is to develop interactive tools that enable the exploration of sequences of events and genomics data. 
Projects such as {\tt cooccurNet}~\cite{ZouWDWWLJP17} provide a basis that can be extended with all-window-length co-occurrence analysis  functionalities. 

\noindent {\bf Mining Problems} In this paper, we expounded co-occurrence counting of itemsets and patterns in the all-window-length analysis model. Going forward, we study mining algorithms in this analysis model, including mining frequent closed  itemsets, i.e. given a sequence $T$ find the top-$k$ itemsets that have highest co-occurrences in an arbitrary window size and for a frequent itemset $X$, there exists no super-pattern $X\subset Y$, with the same co-occurrence as $X$. 
The algorithm requires to mine frequent itemsets for all window lengths in one pass. 
\eat{
\noindent {\bf Extending to Timestamped Sequences} The proposed algorithms operate on a sequence of data points taken at equally spaced points in time. Thus, our sequences are discrete-time data. We plan to study co-occurrence counting and frequent itemset mining in a series of data points indexed in continuous time order. In the continuous setting, we define $T = \{(e,\omega) : e \in I, \omega \in \left[ 0, \tau \right]\}$, where $\tau \geq 0$, to be a set of timestamps in consideration. We can now consider the co-occurrence of items and patterns over an interval of time, which can now be considered as  events not items. We wish to compute the probability of a groups of events happening in time scale $r$:  
\[
    \Pr(I\in \left[ a,b \right],T,r) = \Pr_{y \sim U(\left[ a,b \right])}( \forall e \in I, \exists (e, \omega) \in T : |y-\omega| < r).
\] 
This naturally leads to an analytic definition and suggests a continuous analog of co-occurrence. 
}

\section*{Acknowledgements}
We would like to give special thanks to Lu Zhang and Katherine Seeman for their efforts for the implementation and experimental evaluation of our algorithms. 
\bibliography{references,all}

\begin{thebibliography}{10}

\bibitem{AgrawalIS93}
Rakesh Agrawal, Tomasz Imielinski, and Arun~N. Swami.
\newblock Mining association rules between sets of items in large databases.
\newblock In {\em {SIGMOD}}, pages 207--216, 1993.

\bibitem{AgrawalS94}
Rakesh Agrawal and Ramakrishnan Srikant.
\newblock Fast algorithms for mining association rules in large databases.
\newblock In {\em {VLDB}}, pages 487--499, 1994.

\bibitem{DBLP:conf/kdd/ChangL03}
Joong~Hyuk Chang and Won~Suk Lee.
\newblock Finding recent frequent itemsets adaptively over online data streams.
\newblock In {\em {SIGKDD}}, pages 487--492.

\bibitem{CormodeH08}
Graham Cormode and Marios Hadjieleftheriou.
\newblock Finding frequent items in data streams.
\newblock {\em {PVLDB}}, 1(2):1530--1541, 2008.

\bibitem{DBLP:journals/siamcomp/DatarGIM02}
Mayur Datar, Aristides Gionis, Piotr Indyk, and Rajeev Motwani.
\newblock Maintaining stream statistics over sliding windows.
\newblock {\em {SIAM} J. Comput.}, 31(6):1794--1813, 2002.

\bibitem{DBLP:books/sp/16/DatarM16}
Mayur Datar and Rajeev Motwani.
\newblock The sliding-window computation model and results.
\newblock In {\em Data Stream Management - Processing High-Speed Data Streams},
  pages 149--165. 2016.

\bibitem{PMID:18032723}
Xiangjun Du, Zhuo Wang, Aiping Wu, Lin Song, Yang Cao, Haiying Hang, and
  Taijiao Jiang.
\newblock Networks of genomic co-occurrence capture characteristics of human
  influenza a (h3n2) evolution.
\newblock 18(1), January 2008.

\bibitem{DBLP:conf/focs/FlajoletM83}
Philippe Flajolet and G.~Nigel Martin.
\newblock Probabilistic counting.
\newblock In {\em {FOCS}}, pages 76--82, 1983.

\bibitem{Lavaee+:CC19}
Rahman Lavaee, John Criswell, and Chen Ding.
\newblock Codestitcher: inter-procedural basic block layout optimization.
\newblock In {\em Proceedings of the International Conference on Compiler
  Construction}, pages 65--75, 2019.

\bibitem{DBLP:conf/icdm/LeungK06}
Carson~Kai{-}Sang Leung and Quamrul~I. Khan.
\newblock Dstree: {A} tree structure for the mining of frequent sets from data
  streams.
\newblock In {\em {ICDM}}, pages 928--932, 2006.

\bibitem{LevyG14}
Omer Levy and Yoav Goldberg.
\newblock Dependency-based word embeddings.
\newblock In {\em {ACL}}, pages 302--308, 2014.

\bibitem{Liu:2018:TWC:3291291.3291322}
Yumeng~(Lucinda) Liu, Daniel Busaba, Chen Ding, and Daniel Gildea.
\newblock All timescale window co-occurrence: Efficient analysis and a possible
  use.
\newblock In {\em Proceedings of the 28th Annual International Conference on
  Computer Science and Software Engineering}, CASCON '18, pages 289--292,
  Riverton, NJ, USA, 2018. IBM Corp.

\bibitem{MankuM02}
Gurmeet~Singh Manku and Rajeev Motwani.
\newblock Approximate frequency counts over data streams.
\newblock In {\em {VLDB}}, pages 346--357, 2002.

\bibitem{mikolov2013distributed}
Tomas Mikolov, Ilya Sutskever, Kai Chen, Greg~S Corrado, and Jeff Dean.
\newblock Distributed representations of words and phrases and their
  compositionality.
\newblock In {\em Advances in Neural Information Processing Systems}, pages
  3111--3119, 2013.

\bibitem{pennington-socher-manning:2014:EMNLP2014}
Jeffrey Pennington, Richard Socher, and Christopher Manning.
\newblock Glove: Global vectors for word representation.
\newblock In {\em Proceedings of the 2014 Conference on Empirical Methods in
  Natural Language Processing (EMNLP)}, pages 1532--1543, Doha, Qatar, 2014.

\bibitem{rajaraman2014mining}
Anand Rajaraman, Jure Leskovec, and Jeffrey~D. Ullman.
\newblock {\em Mining Massive Datasets}.
\newblock 2014.

\bibitem{Schutze97}
Hinrich {Sch\"{u}tze}.
\newblock {\em Ambiguity Resolution in Language Learning -- Computational and
  Cognitive Models}.
\newblock Number~10 in CSLI Lecture Notes Series. Center for the Study of
  Language and Information, Stanford, California, 1997.

\bibitem{Shapiroe}
Jason~W. Shapiro and Catherine Putonti.
\newblock Gene co-occurrence networks reflect bacteriophage ecology and
  evolution.
\newblock {\em mBio}, 9(2), 2018.

\bibitem{SiddiquiLWKP20}
Tarique Siddiqui, Paul Luh, Zesheng Wang, Karrie Karahalios, and Aditya~G.
  Parameswaran.
\newblock Shapesearch: {A} flexible and efficient system for shape-based
  exploration of trendlines.
\newblock In {\em {SIGMOD}}, pages 51--65, 2020.

\bibitem{YuY0LP15}
Ziqiang Yu, Xiaohui Yu, Yang Liu, Wenzhu Li, and Jian Pei.
\newblock Mining frequent co-occurrence patterns across multiple data streams.
\newblock In {\em {EDBT}}, pages 73--84, 2015.

\bibitem{Zhang+:POPL06}
Chengliang Zhang, Chen Ding, Mitsunori Ogihara, Yutao Zhong, and Youfeng Wu.
\newblock A hierarchical model of data locality.
\newblock In {\em Proceedings of the ACM SIGPLAN-SIGACT Symposium on Principles
  of Programming Languages}, pages 16--29, 2006.

\bibitem{Zhang+:findingreference04}
Chengliang Zhang, Yutao Zhong, Chen Ding, and Mitsunori Ogihara.
\newblock Finding reference affinity groups in trace using sampling method.
\newblock Technical report, Department of Computer Science, University of
  Rochester, 2004.

\bibitem{Zhong+:PLDI04}
Yutao Zhong, Maksim Orlovich, Xipeng Shen, and Chen Ding.
\newblock Array regrouping and structure splitting using whole-program
  reference affinity.
\newblock In {\em Proceedings of the ACM {SIGPLAN} Conference on Programming
  Language Design and Implementation}, pages 255--266, 2004.

\bibitem{ZouWDWWLJP17}
Yuanqiang Zou, Zhiqiang Wu, Lizong Deng, Aiping Wu, Fan Wu, Kenli Li, Taijiao
  Jiang, and Yousong Peng.
\newblock cooccurnet: an r package for co-occurrence network construction and
  analysis.
\newblock {\em Bioinformatics}, 33(12):1881--1882, 2017.

\end{thebibliography}
\newpage
\pagebreak
\section*{Appendix}
\subsection*{Proof of Theorem \ref{thm:hupdates}}
\begin{proof}
        We have a maximal gap for every subset of $I$ containing $T[i]$. This collection of subsets can be written as
    \[
        C = \{A \subseteq I | A = \{T\left[ i \right]\}\cup B, B\subseteq I\setminus\{T\left[ i \right]\}\}.
    \]
    For each $A \in C$, the update to $H$ is ${(-1)}^{|A|+1}$ to $H\left[\tesla(r_i^1(A),i-1)\right]$ as we have found an $A$-gap of size $\tesla(r_{i}^{1}(A),i-1)$ at index $i$. Let $C_{k} = \{A \in C | \tesla(r_i^1(A),i-1) = k\}$. The incremental updates can be expressed by
    \begin{equation}
        H_{i}[k]= H_{i - 1} \left[ k \right] + \sum_{A \in C_{k}} {(-1)}^{|A|+1},
        \label{update}
    \end{equation}
    for each $k$.
     Suppose $T[i] = r_i^{j_0}(I)$ where $j_0 < |I|$, i.e., $T[i]$ is not the item seen furthest in the past most recently. Then there are $\binom{|I| - j }{\ell}$ sets $A \in C$ of length $\ell + 1$ in which $T[i] = r_i^1(A)$. Thus for $k = \tesla(T[i],i-1)$, we have
     \begin{equation}
     H_{i}[k]-  H_{i - 1}\left[ k \right] =  \sum_{A \in C_{k}} {(-1)}^{|A|+1} = \sum_{\ell = 0}^{|I| - j}\binom{|I| - j }{\ell}(-1)^\ell = (1 - 1)^{|I| - j} = 0.
     \label{eq:otherupdate}
     \end{equation}
     Now for every $ j < j_0$ (which means that $j$ never equals $|I|- 1$ in this case), we have that there are $\binom{|I|-j -1}{\ell}$ members $A\in C$ of length $\ell + 2$ in which, $r^j_i(I) = r_i^1(A)$. Therefore for $k = \tesla(r_i^j(I),i-1)$,
     \[
     H_{i}[k]-  H_{i - 1}\left[ k \right] =  \sum_{A \in C_{k}} {(-1)}^{|A|+1} = \sum_{\ell = 0}^{|I| - j - 1}\binom{|I| - j -1}{\ell}(-1)^{(\ell+1)} = -(1 - 1)^{|I| - j-1} = 0.
     \]
     But for $|I| \geq j > j_0$, there are no such sets $A\in C$ in which $r_i^j(I) = r_i^1(A)$, as $T[i]=r_i^{j_0}(I) $ is contained in all $A \in C$.
     
     But if $j_0 = |I|$, i.e., $T[i] = r_i^{|I|}(I)$, then for each $j < |I|-1$, there are again $\binom{|I|-j -1}{\ell}$ members $A\in C$ of length $\ell + 2$ in which, $r^j_i(I) = r_i^1(A)$, so again equation (\ref{eq:otherupdate}) holds for $k = \tesla(r_i^j(I),i-1)$, giving no net updates for such $k$.
     But there is exactly one $A\in C$ in which $r_i^{|I|-1}(I) = r_i^1(A)$, namely, $\{r_i^{|I|-1}(I), T[i]\}$, and there is also one $A\in C$ in which $T[i] = r_i^{|I|}(I) = r_i^1(A)$, which is $\{T[i]\}$. Therefore we have 
    \begin{align*}
        H_{i}[\tesla(r_{i}^{|I|}(I),i-1)] &= H_{i-1}[\tesla(r_{i}^{|I|}(I),i-1)] + 1,\\
        H_{i}[\tesla(r_{i}^{|I|-1}(I),i-1)] &= H_{i-1}[\tesla(r_{i}^{|I|}(I),i-1)] - 1.
    \end{align*}
    \end{proof}
\subsection*{Proof of Theorem \ref{thm:mult}}
\begin{proof}
    Suppose without loss of generality that $X_{i} \subseteq I$. We wish to find $H_{i}-H_{i-1}$.
    Denote
    \[
        X_{i} = \{r^1_{i}(X_{i}),r^2_{i}(X_{i}),\ldots ,r^{|X_{i}|}_{i}(X_{i})\} = \{x_1,x_2,\ldots ,x_{|X_{i}|}\},
    \]
    as $r$ defined before.
    Now let
    \[
        U_{i} = \{A \subseteq I: A = B \cup \{x_{j}\},B\subseteq I \setminus \{x_{j}\}, j \in \left[ |X_{i}| \right] \},
    \]
    which in words, is all subsets of $I$ that contain at least one member of $X_{i}$.
    Observe that
    \begin{align*}
        U_{i}
    &= \bigcup_{j = 1}^{|X_{i}|}\{A\subseteq I : A =B\cup \{x_{j}\}, B \subseteq I\setminus \{x_{j}\}\}. \\
\end{align*}
Now let $K_{j} = \{A\subseteq I : A = B\cup \{x_{j}\}, B \subseteq I\setminus \{x_{j}\}\}$ for all $j$. Therefore $U_{i} =\bigcup_{j = 1}^{|X_{i}|}K_{j}$.

The update rule is known for each $K_{j}$ based on our previous result. The remains the of the proof is as follows. We can leverage the update rule currently known to compute the total update. But the intersection of $K_{j}$'s is non-empty, meaning if we update according to each $K_{j}$, we would be overcounting some members of $U$. Once this is determined, we will find the update rule according for each arbitrary intersection of these $K_{j}$'s, which completes the proof.

Define $\left< \cdot \right>_{i}$ to be a mapping from subsets of $I$ to an integer valued $n$ dimensional vector. $\left< A \right>_{i}^{k}$ is the sum of the number of maximal gaps of length $k$ ending at index $i$ given by even subsets of $A$, minus the sum of the number of maximal gaps of length $k$ ending at index $i$ given by the odd subsets of $A$. Using this new definition, $\left<U_{i}\right>^{k}_{i} = H_{i}[k]-H_{i-1}\left[ k \right]$. We can now appeal to the inclusion exclusion principle to write that

\begin{equation}
    H_{i} - H_{i-1} = \left< U_{i} \right> = \left< \bigcup_{j= 1}^{|X_{i}|}K_{j}\right> = \sum_{J\subseteq \left[|X_{i}|\right]}{(-1)}^{|J| +1}\left< \bigcap_{j \in J}K_{j}\right>.
    \label{updateprinciple}
\end{equation}
The right hand side of the above equality will now be used.

Denote for any $J\subseteq \left[ |X_{i}| \right]$,
\[
    \mathcal{K}_{J}= \bigcap_{j\in J}K_{j}.
\]
Let $\mathcal{X}_{J}$ be the set of members of $X_{i}$ that lie in every member of $\mathcal{K}_{J}$. Observe that
\[
    \mathcal{X}_{J}=\bigcap_{G\in \mathcal{K}_{J}}G.
\]
It also follows that $\mathcal{X}_{J} = \bigcup_{j \in J}\{x_{j}\}$. Moreover, we can write
\[
    \mathcal{K}_{J}= \{A\subseteq I : A = \mathcal{X}_{J}\cup B, B\subseteq I\setminus \mathcal{X}_{J}\}.
\]
For each set $A \in \mathcal{K}_{J}$, there is a corresponding set $A'$ in $\mathcal{K}'_{J} = \{A \subseteq I : A = \{r_{i}^{1}(\mathcal{X}_{J})\}\cup B, B \subseteq I \setminus \mathcal{X}_{J}\}$, in which $\left< A \right> = {(-1)}^{|J|+1}\left< A' \right>$. This correspondence is easy to find. Let $A \in \mathcal{K}_{J}$. Thus $A = \mathcal{X}_{J}\cup B$, for some $B\in I \setminus \mathcal{X}_{J}$. Then the corresponding set $A' \in \mathcal{K}'_{J}$ is $\{r_{i}^{1}(\mathcal{X}_{J})\}\cup B$. This is clear, because items that lie in every $A \in \mathcal{K}_{J}$ that never satisfy $\arg\min_{e \in A}\tesla(e,i-1)$ for all $A$ never contribute towards any updates and hence can be ignored, except they may change the parity of the set and hence change the sign of the update. From here, we can apply the first theorem taking $I$ in that theorem to be $I\setminus \mathcal{X}_{J}$, which gives
\begin{equation}
    \left<\mathcal{K}_{J}\right>_{i}^{k} = {(-1)}^{|J|+1}\left<\mathcal{K}'_{J}\right>_{i}^{k} = {(-1)}^{|J|+1}
    \begin{cases}
        1 \ \text{for}\ k = \tesla(r_{i}^{|I\setminus\mathcal{X}_{J}|}((I\setminus \mathcal{X}_{J})\cup\{x^{*}\}),i-1)\\
        -1\ \text{for}\ k = \tesla(r_{i}^{1}(\mathcal{X}_{J}),i-1)\\
        0\ \text{otherwise},
    \end{cases}
    \label{j_update}
\end{equation}
when $r_{i}^{1}(\mathcal{X}_{J}) = r_{i}^{|I\setminus \mathcal{X}_{J}|+1}((I\setminus \mathcal{X}_{J})\cup\{r_{i}^{1}(\mathcal{X}_{J})\})$. Every update is $0$ otherwise.

Now assume for all $x\in X_{i}$ and $e \in I\setminus X_{i}$, $\tesla(x,i-1) \geq \tesla(e,i-1)$. For if this does not hold for some $x' \in X_{i}$, then by the above, no updates occur due to $x'$, so analysis is the same.

We now wish to compute the right hand side of equation (\ref{updateprinciple}). We can employ equation (\ref{j_update}) for each $\mathcal{K}_{J}$. If $r_{i}^{1}(\mathcal{X}_{J})\neq r_{i}^{|I\setminus \mathcal{X}_{J}|+1}((I \setminus \mathcal{X}_{J})\cup \{r_{i}^{1}(\mathcal{X}_{J})\})$, that is, the first ranked item of $\mathcal{X}_{J}$ is not ranked below all of $I \setminus \mathcal{X}_{J}$, then $\left< \mathcal{K}_{J}\right> = \mathbf{0}$. We claim that the $J$ in which $\left< \mathcal{K}_{J}\right> \neq \mathbf{0}$ are of the following form:

\begin{equation}
    J_{m} = \{|X_{i}| - b: b \in \left[ m \right]\},
\end{equation}
for $0 \leq  m < |X_{i}|$. We first show that if $J \neq J_{m}$ for some $m$, then $\left< \mathcal{K}_{J}\right> = \mathbf{0}$. If $J\neq J_{m}$ for some $m$, then there exists $b_0$ such that $r_{i}^{|X_{i}| - b_0}(X_{i})\notin \mathcal{X}_{J}$, and there is some $b_1$ such that $b_1 > b_0$ and $r_{i}^{|X_{i}|-  b_1}(X_{i}) \in \mathcal{X}_{J}$. Since $b_1 \leq |X_{i}| -1$, $b_0 < |X_{i}|-1$ which gives that $|X_{i}| - b_0 > 1$. Let $c_0$ and $c_1$ be such that $r_{i}^{c_0}( (I\setminus\mathcal{X}_{J})\cup \{r_{i}^{1}(\mathcal{X}_{J})\}) = r_{i}^{|X_{i}|-b_0}(X_{i})$ and $r_{i}^{c_1}( (I\setminus\mathcal{X}_{J})\cup \{r_{i}^{1}(\mathcal{X}_{J})\}) = r_{i}^{1}(\mathcal{X}_{J})$. We have that $c_0 > c_1$. Now since $c_0 \leq |I\setminus \mathcal{X}_{J}| + 1$, $c_1 \neq |I\setminus X_{J}| + 1$. Therefore $r_{i}^{1}(\mathcal{X}_{J}) \neq r_{i}^{|I\setminus X_{J}|+ 1}( (I\setminus\mathcal{X}_{J})\cup \{r_{i}^{1}(\mathcal{X}_{J})\})$, hence $\left< \mathcal{K}_{J}\right> = \mathbf{0}$.

Now suppose that $J = J_{m}$ for some $m$. Let $c_1$ be such that $ r_{i}^{c_1}(I) = r_{i}^{1}(\mathcal{X}_{J})$. We then have that for any $c < c_1$, $c \in J$, moreover, $r_{i}^{c}(I)\notin  (I \setminus \mathcal{X}_{J})\cup \{r_{i}^{1}(\mathcal{X}_{J})\}$. Thus $r_{i}^{1}(\mathcal{X}_{J}) = r_{i}^{|I \setminus \mathcal{X}_{J}| +1}( (I \setminus \mathcal{X}_{J})\cup \{r_{i}^{1}(\mathcal{X}_{J})\})$, for if not, then there would be $c_0 > c_1$ in which $r_{i}^{c_0}(I) \in (I \setminus \mathcal{X}_{J})\cup \{r_{i}^{1}(\mathcal{X}_{J})\}$, a contradiction.

From this, the right hand side of equation (\ref{updateprinciple}) becomes
\begin{equation}
    \sum_{J \subseteq \left[ |X_{i}| \right]}{(-1)}^{|J| + 1}\left< \mathcal{K}_{J}\right> = \sum_{m = 0}^{|X_{i}|-1}{(-1)}^{m}\left<\mathcal{K}_{J_{m}}\right>.
    \label{linear}
\end{equation}
We now have for $J = J_{m}$, $r_{i}^{1}(\mathcal{X}_{J}) = r_{i}^{|X_{i}| -m }(X_{i})$. Also when $m < |X_{i}| -1$, we have that
\begin{equation}
    r_{i}^{|I\setminus\mathcal{X}_{J}|}(I\setminus \mathcal{X}_{J}\cup\{r_{i}^{1}(\mathcal{X}_{J})\}) = r_{i}^{|I\setminus \mathcal{X}_{J}|} (I \setminus \mathcal{X}_{J}) = r_{i}^{|X_{i}| - (m + 1)}(X_{i}).
\end{equation}
But when $m = |X_{i}| - 1$, $J = [|X_{i}|]$, therefore
\begin{equation}
    r_{i}^{|I\setminus\mathcal{X}_{J}|}(I\setminus \mathcal{X}_{J}\cup\{r_{i}^{1}(\mathcal{X}_{J})\}) = r_{i}^{|I\setminus\mathcal{X}_{J}|}(I\setminus \mathcal{X}_{J}) = r_{i}^{|I\setminus X_{i}|}(I \setminus X_{i}).
\end{equation}
Now for $m < |X_{i}| - 1$, we can rewrite equation (\ref{updateprinciple}) to get
\begin{equation}
    \left<\mathcal{K}_{J_{m}}\right>_{i}^{k} = {(-1)}^{m}
    \begin{cases}
        1\ \text{for}\ k = \tesla (r_{i}^{|X_{i}|-(m+1)}(X_{i}),i-1) \\
        -1\ \text{for}\ k = \tesla(r_{i}^{|X_{i}| -m} (X_{i}),i-1)\\
        0\ \text{otherwise}.
    \end{cases}
\end{equation}
Now define
\begin{equation}
    u{(m)}_{i}^{k} =
    \begin{cases}
        1\ \text{for}\ k = \tesla (r_{i}^{|X_{i}|-(m+1)}(X_{i}),i-1) \\
        -1\ \text{for}\ k = \tesla(r_{i}^{|X_{i}| -m} (X_{i}),i-1)\\
        0\ \text{otherwise},
    \end{cases}
\end{equation}
for $m < |X_{i}| -1$
and
\begin{equation}
    u{(|X_{i}| - 1)}_i^k =
    \begin{cases}
        1\ \text{for}\ k = \tesla (r_{i}^{|I \setminus X_{i}|}(I \setminus X_{i}),i-1) \\
        -1\ \text{for}\ k = \tesla(r_{i}^{1} (X_{i}),i-1)\\
        0\ \text{otherwise},
    \end{cases}
\end{equation}
Taking, $u{(m)}_{i} = (u^1_{i}(m),u^2_{i}(m),\ldots,u^n_{i}(m))$, we can write
\begin{equation}
    \sum_{m = 0}^{|X_{i}|-1}{(-1)}^m\left<\mathcal{K}_{J_{m}}\right> = \sum_{m = 0}^{|X_{i}|-1} {(-1)}^m{(-1)}^m u{(m)}_i^k =  \sum_{m = 0}^{|X_{i}|-1} u{(m)}_i^k.
    \label{minus_cancel}
\end{equation}
Observe that
\begin{equation}
    u{(0)}_{i}^{k} + u{(1)}_{i}^{k} =
    \begin{cases}
        1\ \text{for}\ k = \tesla (r_{i}^{|X_{i}|  -2}(X_{i}),i-1) \\
        -1\ \text{for}\ k = \tesla(r_{i}^{|X_{i}|} (X_{i}),i-1)\\
        0\ \text{otherwise}.
    \end{cases}
\end{equation}
Applying this for all $m < |X_{i}| - 1$ gives
\begin{equation}
    \sum_{m = 0}^{|X_{i}| - 2} u{(m)}_{i}^{k}=
    \begin{cases}
        1\ \text{for}\ k = \tesla(r_{i}^{1}(X_{i}), i-1) \\
        -1\ \text{for}\ k = \tesla(r_{i}^{|X_{i}|}(X_{i}), i-1) \\
        0\ \text{otherwise}.
    \end{cases}
\end{equation}
So combining this with $u{(|X_{i}| - 1)}_{i}^{k}$, we get
\begin{equation}
    \sum_{m = 0}^{|X_{i}| - 1} u{(m)}_{i}^{k}=
    \begin{cases}
        1\ \text{for}\ k = \tesla(r_{i}^{|I \setminus X_{i}|}(I \setminus X_{i}), i-1) \\
        -1\ \text{for}\ k = \tesla(r_{i}^{|I|}(I), i-1) \\
        0\ \text{otherwise},
    \end{cases}
    \label{finish}
\end{equation}
since $r_{i}^{|X_{i}|}(X_{i}) = r_{i}^{|I|}(I)$.
Combining equation (\ref{finish}) with equations (\ref{minus_cancel}), (\ref{linear}), and (\ref{updateprinciple}) (and considering the components of each of those equations), we finally get,
\begin{equation}
    H_{i}\left[ k \right] - H_{i-1}\left[ k \right] =
    \begin{cases}
        1\ \text{for}\ k = \tesla(r_{i}^{|I \setminus X_{i}|}(I \setminus X_{i}), i-1) \\
        -1\ \text{for}\ k = \tesla(r_{i}^{|I|}(I), i-1) \\
        0\ \text{otherwise},
    \end{cases}
\end{equation}
proving the result ($X_{i} = A$ in the statement of the theorem).
\end{proof}

\subsection*{Algorithms}
Here we place pseudocode for the preliminary algorithms mentioned in the main text.
\pagebreak

\begin{algorithm}[H]\SetAlgoLined
\DontPrintSemicolon
\LinesNumbered
\KwIn{Trace $T$, Itemset $I$}
\KwResult{Histogram $H$}
    \tcp*[h]{Map from each set in the power set of $I$ to beginning of last gap}\;
    gap $\leftarrow$ empty map\;
    $H \leftarrow$ empty histogram\;
    \For{$A \subseteq I$}{
        gap$[A]\leftarrow$0\;
    }
    \tcp*[h]{Read through entire trace}\\
    \For{$i=0$ to $n-1$}{
        current$\leftarrow T[i]$\;
        \tcp*[h]{When item is seen update each gap}\;
        \If{current $\in I$}{
            \For{$A \subseteq I\ \wedge$ current $\in A$}{
            \label{line:subset}
                $\tesla \leftarrow i - 1 -\text{gap}[A]$\;
                gap$[A]$ $\leftarrow$ i\; $H[\tesla] \leftarrow H[\tesla] + (-1)^{|A|+1}$\;
        }
        }
    }
    \tcp*[h]{Handle gaps remaining at the end of the trace}\\
    \For{$A \subseteq I\ \wedge$ current $\in A$}
    {
        $\tesla \leftarrow n - 1 - \text{gap}[A]$\;
        $H[\tesla] \leftarrow H[\tesla] + (-1)^{|A|+1}$\;
    }
    \Return{$H$}\;
    \caption{{\tt GAPCOUNTING} }
    \label{algnaive}
    \end{algorithm}

Here we present the pseudocode for the pattern co-occurrence analysis. 
    
    \begin{algorithm}\SetAlgoLined
\DontPrintSemicolon
\LinesNumbered
\KwIn{Trace $T$, ItemSet $I$}
\KwResult{Co-occurrence of all window lengths}
    $H \leftarrow$ empty histogram, $cooc \leftarrow$ [], $\mathcal{S}\leftarrow$ empty book-stack\\;
    \For{$e \in I$}{
        $\mathcal{S}$ += (e, $-\infty$)\;
    }
    $p \gets |I|$,  $m \gets 1$\;
    \While{$\mathcal{S}.{\tt retrieve}(|I|) = \mathcal{S}.{\tt retrieve}(|I|-m)$}{
        $m \gets m+1$
    }
    $p \gets |I| - m$\;
    \For{$i=0$ to $n-1$}{
        $C \leftarrow \{e \in I\,|\,e[0] = T[i-|e|+1| ... e[|e|-1] = T[i]\}$\;
        $\text{min} \gets i$
        \For{$c \in C$}{
            \If{$\mathcal{S}.{\tt find}(c) < \text{min}$}{$\text{min} \gets \mathcal{S}.{\tt find}(c)$
            }
        }
\If{$\text{min} >  \mathcal{S}.{\tt retrieve}(p+1)$}{
                $f \gets i - \mathcal{S}.{\tt retrieve}(|I|)$, $s \gets i - \mathcal{S}.{\tt retrieve}(p+1)$\;
                $H[f] \gets H[f] + 1$, $H[s] \gets H[s] - 1$\;
            }
        \For{$\text{current} \in C$}{
            j $\leftarrow \mathcal{S}.{\tt find}(\text{current})$\;
        $\mathcal{S}.{\tt update}(j)$\;
        \tcp*[h]{Maintain partition}\\
            \If{$j = p = |I|$}{
                $m \gets 2$\;
                \While{$\mathcal{S}$.{\tt retrieve}(|I|-1) = $\mathcal{S}$.{\tt retrieve}(|I|-m)}{
                    $m \gets m+1$
                }
                $p \gets |I| - m$
            }
        \If{$p \leq j < |I|$}{
            $p \gets p + 1$
        }
        }
        }
        \tcp*[h]{Final gap from bottom of book-stack}\\
        $f \gets i- \mathcal{S}$.{\tt retrieve}(|I|)\;
        $H[f] \gets H[f]+1$\;
        \For{$x = 0$ to $|T|-|I|+1$}{
            $S_x \gets 0$\;
            \For{$k = x$ to $|T|$}{
                $S_x \leftarrow S_x + (k-x+1)H[k]$\;
            }
            $cooc[x] \leftarrow (|T| - x + 1) - S_x$\;
        }
    \Return{$cooc$}\;
    \caption{{\tt PAWLCO}}
    \label{algmult}
    \end{algorithm}

\end{document}